\documentclass[conference]{IEEEtran}
\IEEEoverridecommandlockouts
\usepackage{amsmath,amsfonts}
\usepackage{algorithmic}
\usepackage{array}
\usepackage[caption=false,font=normalsize,labelfont=sf,textfont=sf]{subfig}
\usepackage{textcomp}
\usepackage{stfloats}
\usepackage{url}
\usepackage{verbatim}
\usepackage{graphicx}
\usepackage{cite,color}
\usepackage{stmaryrd}
\usepackage[ruled,linesnumbered]{algorithm2e}
\usepackage{amsmath,amssymb,amsfonts}
\usepackage{amsmath}
\usepackage{amsthm}
\usepackage{amssymb}  
\allowdisplaybreaks

\usepackage{pifont}
\newtheorem{lemma}{Lemma}

\newtheorem{theorem}{Theorem}

\usepackage[font={small,it}]{caption}
\begin{document}

\title{Error Analysis for Over-the-Air Federated Learning under Misaligned and Time-Varying Channels}

\author{Xiaoyan Ma\IEEEauthorrefmark{1}, Shahryar Zehtabi\IEEEauthorrefmark{1}, Taejoon Kim\IEEEauthorrefmark{2} and Christopher G. Brinton\IEEEauthorrefmark{1}\\
\IEEEauthorblockA{\IEEEauthorrefmark{1}Electrical and Computer Engineering, Purdue University, IN 47906, USA\\
\IEEEauthorrefmark{2}Electrical, Computer and Energy Engineering, Arizona State University, AZ 85287, USA}
E-mail:\texttt{\{ma946,szehtabi,cgb\}@purdue.edu},
\texttt{taejoonkim@asu.edu}
}
\maketitle
\makeatletter
\def\@IEEEpubidpullup{3\baselineskip}  
\makeatother

\IEEEpubid{%
  \begin{minipage}[t]{\columnwidth}\raggedright
 \vspace{-2.5 em}
  \fontsize{7}{7.2}\selectfont     
    \setlength{\parskip}{0pt}
    \setlength{\baselineskip}{7.2pt} 
    This work was supported in part by the National Science Foundation (NSF) under grants CNS-2146171 and ITE-2326898, and by the Office of Naval Research (ONR) under grant N00014-21-1-2472.
  \end{minipage}\hspace{\columnsep}\makebox[\columnwidth]{}%
}

\begin{abstract}
This paper investigates an OFDM-based over-the-air federated learning (OTA-FL) system, where multiple mobile devices, e.g., unmanned aerial vehicles (UAVs), transmit local machine learning (ML) models to a central parameter server (PS) for global model aggregation. The high mobility of local devices results in imperfect channel estimation, leading to a misalignment problem, i.e., the model parameters transmitted from different local devices do not arrive at the central PS simultaneously. Moreover, the mobility introduces time-varying uploading channels, which further complicates the aggregation process. All these factors collectively cause distortions in the OTA-FL training process which are underexplored. To quantify these effects, we first derive a closed-form expression for a single-round global model update in terms of these channel imperfections. We then extend our analysis to capture multiple rounds of global updates, yielding a bound on the accumulated error in OTA-FL. We validate our theoretical results via extensive numerical simulations, which corroborate our derived analysis.
\end{abstract}
\IEEEpubidadjcol


\vspace{-0.08in}
\section{Introduction}
\IEEEPARstart{W}{ith} the exponential growth of wireless Internet of Things (IoT) devices, massive volumes of data are being generated and stored at edge devices \cite{brinton2025key}. Traditional machine learning (ML) approaches rely on centralized training, where all local devices upload their raw data to a central server, which is often impractical due to limited communication resources and privacy concerns. To address these challenges, federated learning (FL) has emerged as a decentralized ML paradigm that enables local devices to collaboratively train a global model without sharing their raw data, thereby preserving privacy and reducing communication overhead \cite{10706930}.

To enhance the scalability of FL with respect to the number of participating devices, over-the-air computation (AirComp) has been proposed for uplink local model transmission \cite{10538293}. By leveraging the inherent signal superposition property of the wireless multiple access channel, the model parameters transmitted from different edge devices can be organically aggregated. In \cite{9272666}, it has been verified that AirComp can provide strong noise tolerance and reduced latency compared to conventional orthogonal multiple access (OMA) methods. 

To support reliable over-the-air (OTA) aggregation in broadband systems, orthogonal frequency-division multiplexing (OFDM) has been widely adopted as a physical-layer technique for AirComp \cite{10215506}. As a multi-carrier modulation scheme, OFDM combats frequency-selective fading by dividing a high-data-rate signal into multiple orthogonal subcarriers, resulting in improved aggregated signal quality \cite{wang2024digital}.

Despite these advantages offered by OFDM-based over-the-air FL (OTA-FL), communication errors caused by imperfect channel information can severely degrade the learning performance. In \cite{10261509}, the authors investigate the impacts of channel state information (CSI) uncertainty on an AirComp-enabled FL system. The results show that FL results deteriorate with channel uncertainty. 
 In \cite{9780892}, the authors utilize statistical CSI models and investigate the effects of these imperfect CSIs on OTA-FL convergence properties. Furthermore, the authors in \cite{10272265} design a non-coherent receiver
 at the edge server to reduce the channel estimation overhead. They theoretically analyze the mean squared error (MSE) performance of their proposed scheme and the convergence rate of their FL process. 

A key gap in these existing works is their focus on stationary scenarios. In practical OTA-FL systems, local devices may have high mobility, e.g., in unmanned aerial vehicle (UAV)-assisted or vehicle-assisted FL systems. Mobility adds dynamics to imperfect CSI estimations, which leads to a misalignment problem, i.e.,
the model parameters transmitted from different local devices
cannot arrive at the central parameter server (PS) simultaneously. Moreover, mobility introduces time-varying uploading channels, which further complicates the aggregation process at the PS.

\textbf{Contributions.} To address these challenges, we investigate the effects of misalignment and time variation on the performance of OTA-FL systems. The main contributions of our work can be summarized as follows:

\begin{itemize}
    \item We consider OTA-FL systems in which the uploading channels between local devices and the PS are subject to mobility-induced impairments. We develop a comprehensive channel model to capture these dynamics and integrate it into the FL training process. 
    \item By analyzing the aggregated signal received at the central PS, we first bound the error for a single-round global model update in terms of the channel imperfections. We then extend our analysis to the multi-round case and derive closed form expressions of the accumulated error of the OTA-FL training process.

    \item We provide simulation results to compare the analytical error bound we derived with the real accumulated error across multi-round OTA-FL. Our simulations confirm the accuracy of our derived theoretical analysis.
\end{itemize}


\vspace{0.035in}
\section{System Model} \label{sec:sysmodel}
In this section, we introduce our OTA-FL model. As shown in Fig. \ref{system_model}, each device collects data for local model training and uploads their respective models to the central PS for model aggregation. The mobility of local devices causes imperfect channel estimation, resulting in misaligned and time-varying channels, which ultimately distorts the OTA-FL outcome.

\subsection{OFDM-based Signal Transmission Model} \label{ssec:ofdm}
In FL, the training process is divided into $T$ global aggregation iterations $t = 1,...,T$. For the $t$-th global model aggregation, we assume that each local device $k=1,2,...,K$ has an $M$-dimensional local parameter vector $\boldsymbol{\theta}_k(t)$ expressed as
\begin{equation}
  \boldsymbol{\theta}_k(t) 
= 
\bigl[
   \theta_{k,1}(t),\,
  \theta_{k,2}(t),\,
  \dots,\,
 \theta_{k,M}(t)
\bigr]^T \in \mathbb{R}^{M \times 1}.
\end{equation}
In our paper, OFDM transmission is used to combat frequency-selective facing. We assume that there are $F_{sc}$ subcarriers. We thus need $D = \lceil\frac{M}{2 \cdot F_{\mathrm{sc}}}\rceil $ OFDM symbols to carry all these $M$ parameters. To form $D$ OFDM symbols, we divide the original model parameter set into two parts: 
\begin{equation}\label{Real_part}
     \boldsymbol {\theta}_{k}^{d,re}(t)
=
\bigl[
   \theta_{k,2(d-1)F_{\mathrm{sc}} + 1}(t),\,
    \dots,\,
  \theta_{k,(2d-1)\,F_{\mathrm{sc}}}(t)
\bigr]^{T} ,
\end{equation}
\begin{equation} \label{Image_part}
   \boldsymbol{ \theta}_k^{d,im}(t)
=
\bigl[
  \theta_{k,(2d-1)\,F_{\mathrm{sc}} + 1}(t),\,
  \dots,\,
  \theta_{k,2\,d\,F_{\mathrm{sc}}}(t)
\bigr]^T,
\end{equation}
where $d$ is the symbol index. The $d$-th OFDM symbol on all $F_{sc}$ subcarriers can thus be expressed as 
\begin{equation} \label{OFDM_Symbol}
  \boldsymbol{\theta}_k^d(t)  =
 \boldsymbol{\theta}_k^{d,re}(t)
\;+\;
j\,\boldsymbol{ \theta}_k^{d,im}(t) .
\end{equation}

The components of the OFDM symbol in Eq. \eqref{OFDM_Symbol} can be expressed as
\begin{equation}\label{original_transmitted_signal}
    \theta_{k,i}^{d}(t)
=
 \theta_{k,\,2(d-1)\,F_{\mathrm{sc}} + i}(t)
\;+\;
j\,\theta_{k,\,(2d-1)\,F_{\mathrm{sc}} + i}(t),
\end{equation}
with $i=1,2,...,F_{sc}$. We then convert Eq. \eqref{original_transmitted_signal} from frequency domain to time domain for actual OTA transmission through an $F_{sc}$-point IDFT, which is given as
\begin{equation} \label{time-domain-signal}
    x_{k,i}^d(t)
=
\frac{1}{F_{\mathrm{sc}}}
\sum_{f=1}^{F_{\mathrm{sc}}}
 \theta_{k,f}^d(t)\,
e^{\,j\,2\pi\,\frac{fi}{F_{\mathrm{sc}}}},
\quad
i = 1,\dots,F_{\mathrm{sc}}. 
\end{equation}

Subsequently, we add a cyclic prefic (CP) of length $F_{cp}$ at the beginning of each OFDM symbol, making the final OFDM symbol from Eq. \eqref{time-domain-signal} expressed as
\begin{equation}
\begin{aligned}
    &\resizebox{\linewidth}{!}{$\boldsymbol{ x}_{k}^{d}(t)
\!\!=\!\! \!\Bigl[
   x_{k,F_{\mathrm{sc}} - F_{\mathrm{cp}} + 1}^{d}(t),\;
  \dots,\;
   x_{k,F_{\mathrm{sc}}}^{d}(t),\;
   x_{k,1}^{d}(t),\;
  \dots,\;
   x_{k,F_{\mathrm{sc}}}^{d}(t)
\Bigr]^T$}, 
\end{aligned}  
\end{equation}
where $\boldsymbol{ x}_{k}^{d}(t) \in \mathbb{C}^{(F_{cp}+F_{sc}) \times 1}$.
With the above definition, the final transmitted signal set of each local device can be expressed as
\vspace{-0.15cm}
\begin{equation}
    \boldsymbol{x}_{k}(t) 
=
\begin{bmatrix}
 \boldsymbol{x}_{k}^{1}(t), \,\,
 \boldsymbol{x}_{k}^{2}(t), \,\,
\cdots \,\,,
 \boldsymbol{x}_{k}^{D}(t)
\end{bmatrix}^T.
\end{equation}

\vspace{0.05in}
\subsection{Misaligned Time-Varying Channel Model } \label{ssec:channel}


\begin{figure}[t!]
\centering
\includegraphics[width=0.75\linewidth]{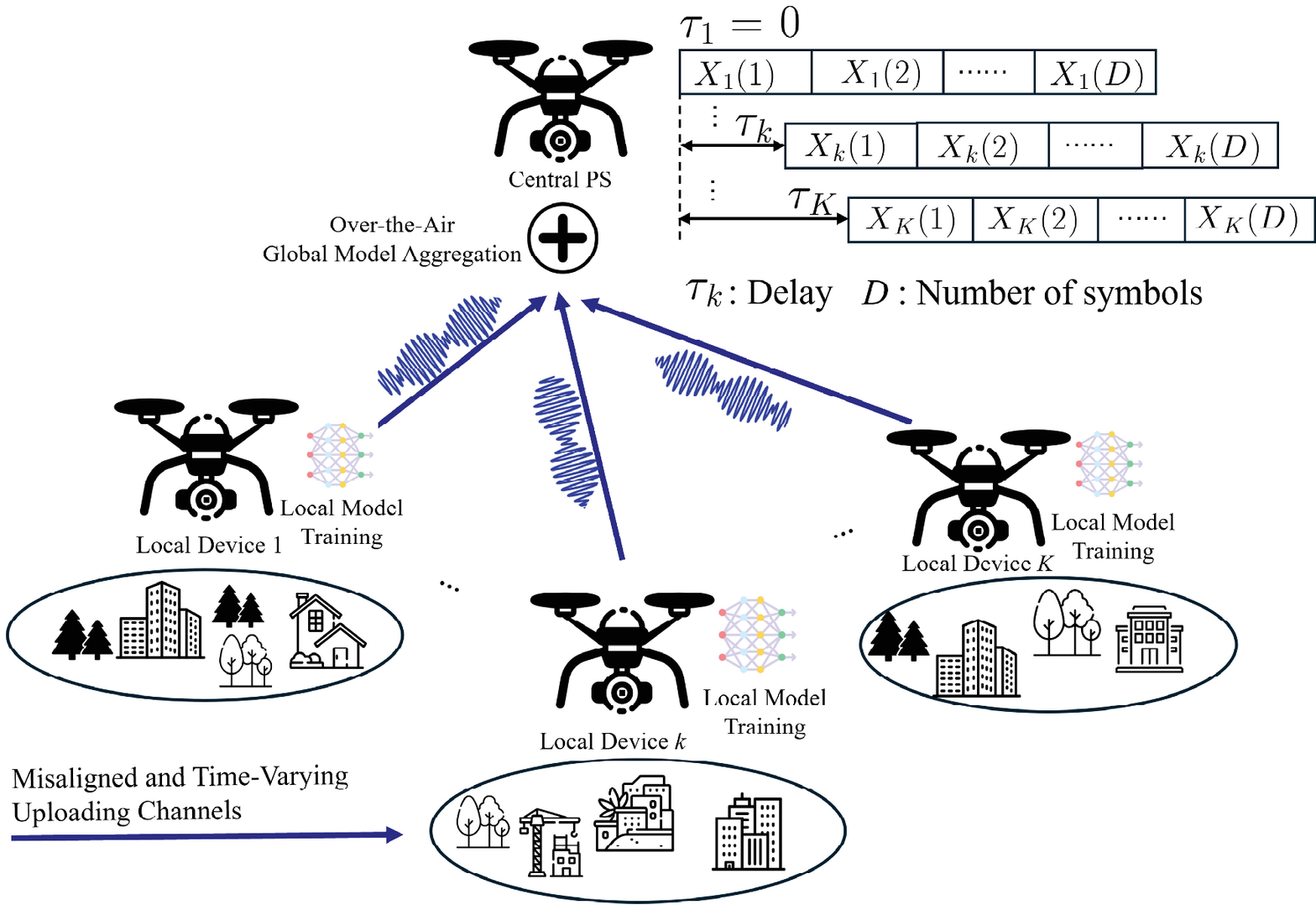}
\caption{Considered OTA-FL system model where each mobile device uploads the respective model to PS for global model aggregation.}
\label{system_model}
\end{figure}

 \begin{figure*}[hb] 

 	\centering
 	\begin{equation} \label{EXPENDED}	\tag{15}
         \begin{aligned}
 		Y_{n,u}^d(t)
&=
\sum_{i=1}^{F_{\mathrm{sc}}}
y_{n,i}^d(t)\,
e^{-\,j\,\frac{2\pi\,u\,i}{F_{\mathrm{sc}}}}
= \frac{1}{F_{\mathrm{sc}}}
\sum_{v=1}^{F_{\mathrm{sc}}}
\sum_{k=1}^{K}
 \theta_{k,v}(t)
\left(
  \sum_{\ell=1}^{L}
  \sum_{i=1}^{F_{\mathrm{sc}}}
  h_{k,n,\ell,i}(t)\,
  e^{-\,j\frac{2\pi (u-v)\,i}{F_{\mathrm{sc}}}}
  \,e^{-\,j\frac{2\pi\,\left(\tau_{kn,l}(t) + \tau_{kn,l}^{d}(t)\right)}{F_{\mathrm{sc}}}}
\right) + Z_{n,u}^d(t), \hspace{-1cm}\\
& =
\sum_{k=1}^K H_{k,n,u,u}^d(t)\, \theta_{k,u}^d(t)
+
\sum_{\substack{v=1 \\ v \neq u}}^{F_{\mathrm{sc}}}
\sum_{k=1}^K
H_{k,n,u,v}^d(t)\, \theta_{k,v}^e(t)
+
Z_{n,u}^d(t). 
\end{aligned}
 	\end{equation}
 \end{figure*}
At the beginning of each global round $t$, the channel from edge local device $k$ to the $n$-th antenna at the central PS can be expressed as \cite{10018930}
\vspace{-0.25cm}
\begin{equation} \label{Chnnel}
h_{k, n}^{d}(t)
= \sum_{l=1}^{L}
    h_{kn,l}^{d}(t)\,\delta\left(\lambda - \tau_{kn,l}(t) - \tau_{kn, l}^{d}(t)\right),
\end{equation}
where $L$ is the number of paths. In Eq. \eqref{Chnnel}, $h_{kn,l}^{d}(t) $ is the channel gain of the $l$-th path. Note that there are two delay parameters in the channel model: (i) $\tau_{kn,l}(t)$ represents the misalignment delay for the $l$-th path from local device $k$ to the $n$-th antenna at the central PS. This is caused by the transmitted signals from different local devices not arriving simultaneously at the central PS due to imperfect channel estimation and imperfect synchronization in the system. (ii) $\tau_{kn, l}^{d}(t)$ is the time delay for OFDM symbol $d$ caused by the time varying channels. We further introduce a parameter $a$ to model this time-varying characteristic of the channels as follows \cite{10018930}:
\begin{equation} \label{time_varying_channel}
    h_{kn,l}^{d}(t)
= \sqrt{1 - \alpha^2}\,h_{kn,l}^{d-1}(t)
+ \alpha c,
\end{equation}
where $c$ is a complex random variable with zero mean
and variance $\sigma_c^2$. We observe that $\alpha$ in Eq. \eqref{time_varying_channel} represents the dependence between two continuous OFDM symbols, and larger $\alpha$ corresponds to faster channel variations.

Then, at the $n$-th receiving antenna, after removing the CP of length $F_{cp}$, the $i$-th entry of the received signal can be expressed as 
\vspace{-0.15cm}
\begin{equation}
    y_{n,i}^d(t)
\!\!=\!\!\!\!\!
\sum_{K=1}^{K}\!
\sum_{\ell=1}^{L}
h_{knl,i}^d(t)\, x^d_{k,\,i -\tau_{kn,l}(t)\! - \!\tau_{kn, l}^{m}(t)}(t)
\!+\!
z_{n,i}^d(t), 
\end{equation}
where $i=1,2,...,F_{sc}$. Finally, through an $F_{sc}$-point DFT, we recover the original signal as
\vspace{-0.1cm}
\begin{equation}\label{recovered_signal}
    Y_{n,u}^d(t)
=
\sum_{i=1}^{F_{\mathrm{sc}}}
y_{n,i}^d(t)\,
e^{-\,j\,\frac{2\pi\,u\,i}{F_{\mathrm{sc}}}} , \,\,\, u=1,2,...,F_{sc}.
\end{equation}

Combining the frequency domain received signals from all $N$ receiving antennas at the parameter server and applying a combining vector, we obtain
\vspace{-0.1in}
\begin{equation}
\hat{Y}_{u}^d(t)=\boldsymbol{w}_u^d(t)\left[
Y_{1,u}^d(t),
Y_{2,u}^d(t),
\cdots,
Y_{N,u}^d(t)
\right]^T,
\end{equation}
where $\boldsymbol{w}_u^d(t)\in \mathbb{C}^{1 \times N}$ is the combining vector on subcarrier $u$ for global round time $t$.
Then the estimated signal at the parameter server can be expressed as 
\begin{equation}
  \hat{\theta}_{2(d-1)F_{\mathrm{sc}} + u}(t)
\!\!=\!\!
\frac{Re({\hat{Y}_{u}^d(t)})}{K} ,\,\,
    \hat{\theta}_{(2d-1)F_{\mathrm{sc}} + u}(t)
\!\!=\!\!
\frac{Im({\hat{Y}_{u}^d(t)})}{K}.
\end{equation}

Compared with the original transmitted signal in Eq. \eqref{original_transmitted_signal}, we can see that the index $k$ for the local device disappears since they are wirelessly combined over the channel.

\subsection{Combining Vector Design at the Central PS} \label{ssec:combination}
For the combining vector design, we need to expand the recovered signal in Eq. \eqref{recovered_signal} to Eq. \eqref{EXPENDED} as shown at the bottom of this page. In Eq. \eqref{EXPENDED}, $H_{k,n,u,v}^d(t)$ represents the effective frequency domain channel from local device $k$ on subcarrier $u$ to the $n$-th antenna at central PS on subcarrier $v$ and can be expressed as
\setcounter{equation}{15}
\vspace{-0.2cm}
\begin{equation} 
\resizebox{\linewidth}{!}{$H_{k,n,u,v}^d(t) \!\!
=\!\!
\frac{1}{F_{\mathrm{sc}}}\!\!
\sum_{\ell=1}^{L}\!
\sum_{i=1}^{F_{\mathrm{sc}}}
h_{k,n,\ell,i}^d(t)
e^{-j\frac{2\pi\bigl(i(u-v) + \tau_{kn,l}(t) + \tau_{kn,l}^{d}(t)\bigr)}{F_{\mathrm{sc}}}}.$}
\end{equation}
Combining all received frequency domain signals from all the $N$ receiving antennas through maximum ratio combining vector, we obtain
\vspace{-0.1in}
\begin{equation} \label{signal_after_combining}
\hat{Y}_u^d(t)
=
\frac{1}{N}
\sum_{n=1}^N
\Biggl(
  \sum_{k=1}^K
  H_{k,n,u,u}^d(t)
\Biggr)^*
\,Y_{n,u}^d(t).
\end{equation}

Here, $
  \sum_{k=1}^K
  H_{k,n,u,u}^d(t)
$ represents all the channels from all local devices to the $n$-th antenna on subcarrier $u$ at the central PS. Consequently, by multiplying the conjugate of $
  \sum_{k=1}^K
  H_{k,n,u,u}^d(t)
$, we achieve maximum ratio combining at the central PS.
\vspace{-0.2cm}
\section{Error Analysis for Global Model Update} \label{sec:analysis}
In this section, we provide detailed analysis of the model aggregation error. We first investigate a single-round model update, then extend to multiple rounds case to evaluate the accumulated error for final OTA-FL outcome.

\subsection{Error Analysis for Single-round Global Model Update} \label{ssec:errorsingle}
To do the error analysis, we need to expand the estimated signal in Eq. \eqref{signal_after_combining}, which can be expressed in detail as
\vspace{-0.1in}

\begin{small}
\begin{align} \label{error0}
& \hat{Y}_u^d(t) =
\underbrace{\frac{1}{N}
\sum_{n=1}^N
\sum_{k=1}^K
\bigl\lvert H^d_{k,n,u,u}(t) \bigr\rvert^2 \, \theta^d_{k,u}(t)}_{\text{(a) desired FL  parameter}}
\nonumber \\
&+
\underbrace{\frac{1}{N}
\sum_{n=1}^N
\sum_{k=1}^K
\sum_{\substack{k'=1 \\ k' \neq k}}^K
\bigl(H^d_{k,n,u,u}(t)\bigr)^*
\,H^d_{k',n,u,u}(t)\,
 \theta^d_{k,u}(t)\, \theta^d_{k',u}(t)}_{\text{(b) interference among same subcarrier}}
\nonumber \\
&+
\underbrace{\frac{1}{N}
\sum_{n=1}^N
\sum_{\substack{v=1 \\ v \neq u}}^{F_{\mathrm{sc}}}
\sum_{k=1}^K
\bigl(H^d_{k,n,u,u}(t)\bigr)^*
\,H^d_{k,n,u,v}(t)\,
 \theta^d_{k,u}(t)\, \theta^d_{k,v}(t)}_{\text{(c) ICI term where }v \neq u}
\nonumber \\
&+
\underbrace{\frac{1}{N}
\sum_{n=1}^N
\sum_{k=1}^K
\bigl(H^d_{k,n,u,u}(t)\bigr)^*
\,Z^d_{n,u}(t)}_{\text{(d) noise term}}.
\end{align}
\end{small}
\vspace{-0.1in}

Here $(a)$ is the desired FL parameter generated through OTA aggregation, $(b)$ is the interference among same subcarrier $u$, $(c)$ is the inter-carrier interference (ICI) betweem subcarriers $u$ and $v$, and $(d)$ is the noise term. We want to give an exact expression of the error terms $(b)$, $(c)$ and $(d)$ in Eq. \eqref{error0}, thereby we analyze them separately. 

\subsubsection{\textbf{Interference among same subcarrier}} \label{sssec:error_b}
First, for (b) in Eq. \eqref{error0}, i.e., interference among same subcarrier, we introduce a parameter
\begin{equation}
    \kappa^{d}_{j,u}(t) = \frac{1}{N} \sum_{n=1}^{N} \sum_{\substack{k=1 \\ k \ne j}}^{K} \left( H^{d}_{k,n,u,u}(t) \right)^{*} H^{d}_{j,n,u,u}(t).
\end{equation}

The key idea behind defining $\kappa^{d}_{j,u}(t)$ is to enable calculation of the correlation between two channels, so that we can quantify the interference from the other users to user $j$ on subcarrier $u$ for the $d$-th symbol in global round time $t$. Then we need to calculate the expectation of $| \kappa^{d}_{j,u}(t) |^2$, i.e., $\mathbb{E} [ | \kappa^{d}_{j,u}(t) |^2 ] $. To achieve this, we further define 
\vspace{-0.08in}
\begin{equation}
    X_{k,j,n,u}^d = \left( H^{d}_{k,n,u,u}(t) \right)^{*} H^{d}_{j,n,u,u}(t) ,
\end{equation}
 then, we have 
 \vspace{-0.2in}
 
 \begin{small}
\begin{equation}
    \begin{aligned}
        & \mathbb{E} \left[ \left| \kappa^{d}_{j,u}(t) \right|^2 \right]
        = \mathbb{E} \Big[ \big| (1/N) \sum_{n=1}^{N} \sum_{k=1, k \ne j}^{K} X_{k,j,n,u}^d \big|^2 \Big] \\
        &\leq \frac{1}{N^2} N (K - 1) \mathbb{E} \left[ |X_{k,j,n,u}^d|^2 \right]
        = \frac{K - 1}{N} \mathbb{E} \left[ |X_{k,j,n,u}^d|^2 \right] .
    \end{aligned}
\end{equation}
 \end{small}
 \vspace{-0.05in}
 
To give an exact expression of $\mathbb{E} [ |X_{k,j,n,u}^d|^2 ] $, we need the following statistical assumptions:
\begin{itemize}
    \item Each channel coefficient $h_{kn,l}^{d}(t)$ is independent and follows the same Rayleigh fading distribution, i.e.,   
    $h_{kn,l}^{d}(t) \sim \mathcal{CN}(0, \sigma_h^2) $;
    
    \item We use $\tilde{\tau}_{kn,l}^d(t)=\tau_{kn,l}(t) + \tau_{kn,l}^{d}(t)$ to represent the overall delay for local device $k$ at the $d$-th OFDM symbol, then the delay difference between local device $k$ on path $l$ and local device $j$ on path $l'$ for the $d$-th OFDM symbol can be expressed as  $\dot{\tau}^{d}_{k,j,l,l'}(t)=\tilde{\tau}_{kn,l}^d(t)- \tilde{\tau}^{d}_{jn,l'}(t)$. We assume that for different user pair $(k,j)$, $\dot{\tau}^{d}_{k,j,l,l'}(t)$ are independent, but they follow the same Gaussian distribution with mean $\mu_{\tau}$ and variance $\sigma_{\tau}^2$.
\end{itemize}
Based on the above statistical assumptions, we have the following lemma regarding to the upper bound of the interference term among the same subcarrier, i.e., term $(b)$ of Eq. \eqref{error0}:
\begin{lemma} \label{lemma:interference}
The upper bound of the interference term among the same subcarrier can be expressed as
\begin{equation} \label{eqn:error_b}
    \mathbb{E} \left[\! \left| \kappa^{d}_{j,u}(t) \right|^2 \!\right] \!\!= \!\!\frac{K - 1}{N}  L^2 \sigma_h^4 \frac{1}{\sqrt{1 + \frac{4 \pi^2 \sigma_\tau^2}{F_{sc}^2}}}  \frac{1}{1 + \left( \frac{2 \pi \mu_\tau}{F_{sc}} \right)^2}.
\end{equation}
\end{lemma}

\begin{proof}
According to Parseval's Theorem, if the time-domain channel has variance $\sigma_h^2$, then each frequency-domain subcarrier has variance $L\sigma_h^2$ after DFT. Based on this, we have
\vspace{-0.12in}

\begin{small} 
\begin{equation} \label{b_analysis}
\begin{aligned}
 & \mathbb{E} \left[ |X_{k,j,n,u}^d|^2 \right] 
 = \mathbb{E} \left[\left |\left( H^{d}_{k,n,u,u}(t) \right)^{*} H^{d}_{j,n,u,u}(t) \right|^2 \right]  \\
&\resizebox{\linewidth}{!}{$\overset{(1)}{=} \left( L \sigma_h^2 \right)^2 \mathbb{E} \left[\eta_b(\mu_{\tau},\sigma^2_{\tau})\right] \overset{(2)}{\leq} L^2 \sigma_h^4 
\frac{1}{\sqrt{1 + \frac{4 \pi^2 \sigma_\tau^2}{F_{sc}^2}}}  \frac{1}{1 + \left( \frac{2 \pi \mu_\tau}{F_{sc}} \right)^2}.$} \hspace{-1cm}
\end{aligned}
\end{equation}
\end{small} 
\vspace{-0.1in}

In Eq. \eqref{b_analysis}, $(1)$ holds due to the fact that the channel coefficients $h_{kn,l}^{d}(t)$ in Eq. \eqref{Chnnel} are independent and follow Guassian random distribution with variance $\sigma_h^2$, and we define $\eta_b(\mu_{\tau},\sigma^2_{\tau})=\text{sinc}^2 
(
(2\pi\dot{\tau}^{d}_{k,j,l,l^{'}}(t)) / F_{sc}
)
$ to approximate the timing errors in OFDM systems \cite{1599595}, and $(2)$ holds due to the fact that $\text{sinc}^2(x) \leq \frac{1}{1+(\pi x)^2}$, then we have
\vspace{-0.15in}

\begin{small} 
\begin{equation}
    \mathbb{E} \left[ \text{sinc}^2\left( \frac{2\pi\dot{\tau}^{d}_{k,j,l,l^{'}}(t)}{F_{sc}} \right) \!\right] 
\!\! \leq \!
\mathbb{E}\!\!\left[\!\! \left(\!\! 1 \!+\! \left( \frac{2\pi\dot{\tau}^{d}_{k,j,l,l^{'}}(t)}{F_{sc}}\! \right)^2 \right)^{-1}\! \right],
\end{equation}
\end{small}
\vspace{-0.2in}

Since $\dot{\tau}^{d}_{k,j,l,l^{'}}(t) \backsim \mathcal{CN}(\mu_{\tau}, \sigma_{\tau}^2) $, we further have
\vspace{-0.08in}

\begin{small} 
\begin{equation}
    \mathbb{E}\!\! \left[ \!\! \left( 1 \!\! +\!\!  \left( \frac{2\pi\dot{\tau}^{d}_{k,j,l,l^{'}}(t)}{F_{sc}} \right)^2 \!\right)^{-1} \right]\!\! =\!\! \frac{1}{\sqrt{1 + \frac{4 \pi^2 \sigma_\tau^2}{F_{sc}^2}}}  \frac{1}{1 + \left( \frac{2 \pi \mu_\tau}{F_{sc}} \right)^2}.
\end{equation}
\end{small}
\vspace{-0.1in}

Combining all the above analyses, we can obtain the upper bound of term $(b)$ as shown in Lemma \ref{lemma:interference}.
\end{proof}

\subsubsection{\textbf{Inter-carrier interference}} \label{sssec:error_c}
For the ICI term shown in part $(c)$ of Eq. \eqref{error0}, we also define a parameter $\zeta_{u,v}^d(t)$ to measure the interference from other subcarriers on subcarrier $u$ as shown below:
\vspace{-0.12in}

\begin{small}
\begin{equation} \label{ICI}
\zeta_{u}^d(t) = \frac{1}{N} \sum_{n=1}^{N} \sum_{\substack{v=1 \\ v \ne u}}^{F_{sc}} \sum_{k=1}^{K} \left( H_{kn,uu}^d(t) \right)^* H_{kn,uv}^d(t).
\end{equation}
\end{small}
\vspace{-0.12in}

Note that for OFDM systems, most of the ICI leakage comes from neighboring subcarriers, and the leakage decays rapidly as the subcarriers move further away from subcarrier $u$. Based on this observation, we consider a window of length $2q$ around the subcarrier $u$. The value of $q$ is related to time-varying channel parameter $\alpha$ as we defined in
 Eq. \eqref{time_varying_channel}. Normally, larger $\alpha$ will lead to more serious ICI and generate higher $q$ value
 \cite{10018930}. Now, instead of summing over all $v \in [1, F_{sc}]$ and $v \neq u$ in Eq. \eqref{ICI}, we limit the sum to the above-mentioned window. Specifically, we only focus on 
\vspace{-0.08in}
 \begin{equation}
       v \in [u-q,..,u-1] \cup [u+1,...,u+p].
 \end{equation}
 \vspace{-0.1in}
With this, we can redefine  Eq. \eqref{ICI} as
\begin{equation} \label{ICI1}
\zeta_{u}^d(t) = \frac{1}{N} \sum_{n=1}^{N} \sum_{\substack{v=u-q \\ v \ne u}}^{v=u+q} \sum_{k=1}^{K} \left( H_{kn,uu}^d(t) \right)^* H_{kn,uv}^d(t).
\end{equation}
 
To give a closed-form expression for Eq. \eqref{ICI1}, we first focus on one subcarrier $v$. We have
\begin{equation} \label{ICI_one}
\zeta_{u,v}^d(t) = \frac{1}{N} \sum_{n=1}^{N} \sum_{k=1}^{K} \left( H_{kn,uu}^d(t) \right)^* H_{kn,uv}^d(t),
\end{equation}
Then the total ICI leakage into subcarrier $u$ can be shown as
\vspace{-0.in}
\begin{equation} \label{ICI_all}
\zeta_{u}^d(t) =  \sum_{\substack{v=u-q, v \ne u}}^{v=u+q} \zeta_{u,v}^d(t) .
\end{equation}

Note that the ICI distribution is non-uniform, accounting for the fact that closer subcarriers (e.g., $u-1$ and $u+1$)
leak more ICI into subcarrier 
$u$, while farther ones (e.g., $u-q$ and $u+q$) leak less. We define the ICI leakage from subcarrier $v$ as a decaying function of its distance from another subcarrier $u$ using the following distance-based weight function:
\vspace{-0.1 in}

\begin{small}
\begin{equation}\label{weight}
    w(\Gamma)=\frac{1}{\Gamma^\gamma} 
    \Bigg/
    \left(\sum_{i=1}^q \frac{2}{i^\gamma} \right),
\end{equation}
\end{small}
\vspace{-0.1 in}

Here, $\Gamma=|v-u|$ and $\gamma$ represents the decay factor, which controls how quickly the ICI decays. With Eq. \eqref{weight}, we can assign larger weights to closer subcarriers than to farther subcarriers. 
Then we need to calculate the weighted ICI leakage as
\vspace{-0.1in}
\begin{equation} \label{568}
\begin{aligned}
\mathbb{E} \left[ \left| \zeta_{u}^d(t) \right|^2 \right]
&=\sum_{v=u-q,v \neq u}^{u+q} w(\Gamma) \mathbb{E} \left[ \left| \zeta_{u,v}^d(t) \right|^2 \right]. \\    
\end{aligned}
\end{equation}
\vspace{-0.05in}
And we have the following Lemma \ref{lemma:ici} regarding the upper bound of Eq. \eqref{568}.
\begin{lemma} \label{lemma:ici}
 The ICI term can be finally bounded by 
\vspace{-0.05in}
\begin{equation} \label{ICI13}
\begin{aligned}
   & \mathbb{E} \left[ \left| \zeta^{d}_{u,v}(t) \right|^2 \right] 
    = \mathbb{E} \left[ \zeta^{d}_{u,v}(t) \cdot \left( \zeta^{d}_{u,v}(t) \right)^* \right]\\
    &\resizebox{\linewidth}{!}{$=\!\! \frac{1}{N^2}\!\! \sum_{n=1}^{N} \!\sum_{k=1}^{K} \!\sum_{n'=1}^{N} \!\sum_{k'=1}^{K} \!\!
\mathbb{E} \left[ \left( H^d_{kn,uu} \right)^*\!\! H^d_{kn,uv}  H^d_{k'n',uu}\!\!\left( H^d_{k'n',uv} \right)^* \right]$}\\
&=\frac{K \sigma_h^4}{N} \cdot \sum_{\Gamma=1}^{q} 2 w(\Gamma) \frac{1}{\sqrt{1 + \frac{4 \pi^2 \Gamma^2\sigma_\tau^2}{F_{sc}^2}}}  \frac{1}{1 + \left( \frac{2 \pi \Gamma^2\mu_\tau}{F_{sc}} \right)^2}.
\end{aligned}
\end{equation}
 \vspace{-0.25in}
 \end{lemma}

 \begin{proof}
First, $\mathbb{E} \left[ \left| \zeta_{u,v}^d(t) \right|^2 \right]$ can be derived as 
\vspace{-0.1in}
\begin{equation} \label{ICI12}
\begin{aligned}
   & \mathbb{E} \left[ \left| \zeta^{d}_{u,v}(t) \right|^2 \right] 
    = \mathbb{E} \left[ \zeta^{d}_{u,v}(t) \cdot \left( \zeta^{d}_{u,v}(t) \right)^* \right]\\
    &\resizebox{\linewidth}{!}{$= \frac{1}{N^2} \sum_{n=1}^{N} \sum_{k=1}^{K} \sum_{n'=1}^{N} \sum_{k'=1}^{K} \mathbb{E} \left[ \left( H^d_{kn,uu} \right)^* H^d_{kn,uv}  H^d_{k'n',uu}\left( H^d_{k'n',uv} \right)^* \right].$}
\end{aligned}
\end{equation}

And Eq. \eqref{ICI12} can be divided into two cases:
\addtolength{\topmargin}{0.012in}
\begin{itemize}
    \item When $(k,n) \neq (k^{'},n^{'})$, we have 
    \vspace{-0.1in}
    \begin{equation} 
        \mathbb{E} \left[ \left( H^d_{kn,uu} \right)^* H^d_{kn,uv}  H^d_{k'n',uu}\left( H^d_{k'n',uv} \right)^* \right]=0,
    \end{equation}
    since we assume that all the channel coefficients are independent complex Gaussian random variables.

    \item When $(k,n) = (k^{'},n^{'})$, we have \vspace{-0.1 in}

      \begin{equation} 
      \begin{aligned}
          \mathbb{E} \left[ \left( H^d_{kn,uu} \right)^* H^d_{kn,uv}  H^d_{k'n',uu}\left( H^d_{k'n',uv} \right)^* \right]\\
        = \sigma_h^4\mathbb{E} \left[ \text{sinc}^2\left( \frac{2\pi (u-v)\tilde{\tau}_{kn,l}^{d}(t)}{F_{sc}} \right) \right], \end{aligned} \end{equation}

    \vspace{-0.1 in}
        
        using the assumption that $h_{kn,l}^{d}(t)\sim \mathcal{CN}(0, \sigma_h^2)$ and all channel coefficients are independent.
\end{itemize}
Based on the above two cases, we can derive the upper bound of the ICI term as shown in Lemma \ref{lemma:ici}.
\end{proof}
 \vspace{-0.1in}
\subsubsection{\textbf{Noise term}} \label{sssec:error_d}
Finally, for the noise term given in part (d) of Eq. \eqref{error0}, we have 
\vspace{-0.15in}

 \begin{small}
\begin{equation} \label{eqn:error_d}
     \mathbb{E} \left[ \left|  \frac{1}{N}
\sum_{n=1}^N
\sum_{k=1}^K
\bigl(H^d_{k,n,uu}(t)\bigr)^*
\,Z^d_{n,u}(t) \right| ^2 \right]
=\frac{K\sigma_{h}^2 \sigma_{z}^2}{N}.
\end{equation}
\end{small}

We simplified each term in Eq. \eqref{error0} in Secs. \ref{sssec:error_b}, \ref{sssec:error_c} and \ref{sssec:error_d}. Then we have the following theorem:

\begin{theorem} \label{thm:theorem}
The overall error $e(t)$ for the $t$-th global model update can be bounded as
\vspace{-0.17in}

\begin{small}
\begin{equation}
\begin{aligned}\label{error}
    & e(t)=\frac{K - 1}{N}  L^2 \sigma_h^4 \frac{1}{\sqrt{1 + \frac{4 \pi^2 \sigma_\tau^2}{F_{sc}^2}}}  \frac{1}{1 + \left( \frac{2 \pi \mu_\tau}{F_{sc}} \right)^2} \\
    &\resizebox{\linewidth}{!}{$+\frac{K \sigma_h^4}{N} \!\!\sum_{\Gamma=1}^{q} 2 w(\Gamma) \frac{1}{\sqrt{1 + \frac{4 \pi^2 \Gamma^2\sigma_\tau^2}{F_{sc}^2}}}  \frac{1}{1 + \left( \frac{2 \pi \Gamma^2\mu_\tau}{F_{sc}} \right)^2} \!+\!\frac{K\sigma_{h}^2 \sigma_{z}^2}{N}.$}
\end{aligned}
\end{equation}
\end{small}
\end{theorem}
\vspace{-0.15in}
\begin{proof}
    Follows from Lemmas 1\&2 and Eq. \eqref{eqn:error_d}.
\end{proof}


\subsection{Accumulated Error}\label{ssec:errormulti}
In Theorem \ref{thm:theorem}, we derived the distortions caused by OTA computation during the $t$-th round of global model aggregation. Specifically, the $t$-th round global model aggregation at the central PS can be expressed as
    \begin{equation}
        \boldsymbol{\theta}(t)=\frac{1}{K}\sum_{k=1}^K \boldsymbol{\theta}_k(t) + \boldsymbol{\epsilon}(t).
    \end{equation}

Here, $\boldsymbol{\theta}(t)$ denotes the global model generated at the PS in the $t$-th round, $\boldsymbol{\theta}_k$ is the local model at device $k$, and $\boldsymbol{\epsilon}(t)$ is a Gaussian random vector that represents the distortion caused by the misalignment and the time-varying channels, satisfying $e(t)=\mathbb{E}[
\| \boldsymbol{\epsilon}(t) \|^2]$ from Theorem \ref{thm:theorem}. Then, the PS broadcasts the
updated global model back to all local devices to facilitate the next round of updates. The $(t+1)$-th local model update at local device $k$ can be shown as
    \begin{equation} \label{local_model_update}
        \boldsymbol{\theta}_{k}(t+1)= \boldsymbol{\theta}(t) - \beta \nabla F(\boldsymbol{\theta}(t)).
    \end{equation}

In Eq. \eqref{local_model_update}, $\beta$ represents the learning rate and $ \nabla F(\boldsymbol{\theta}(t))$ denotes the local gradient. 
Then, the $(t+1)$-th global model aggregated at the central PS can be expressed as
\vspace{-0.22 in}

\begin{small}
    \begin{align}
        \boldsymbol{\theta}(t+1)&= \frac{1}{K}\sum_{k=1}^K \boldsymbol{\theta}_{k}(t+1) + \boldsymbol{\epsilon}(t+1)   \\
       &= \boldsymbol{\theta}(t) -\beta \boldsymbol{g}(t)+\beta\nabla F(\boldsymbol{\epsilon}(t))+\boldsymbol{\epsilon}(t)+\boldsymbol{\epsilon}(t+1).
    \end{align}
\end{small}
\vspace{-0.2 in}

Here, $\boldsymbol{g}(t)$ is the true gradient without errors, i.e., $\nabla{F}(\boldsymbol{\theta}(t)) = \boldsymbol{g}(t) + \nabla{F}(\boldsymbol{\epsilon}(t))$. Based on the above analyses, we have the following recursive model for the final global model $\boldsymbol{\theta}(T)$: 
\vspace{-0.2in}

\begin{small}
\begin{equation} \label{accumulated_error}
    \boldsymbol{\theta}(T)\!=\! \boldsymbol{\theta}(0) -\beta \sum_{t=0}^{T-1}\boldsymbol{g}(t)+\sum_{t=0}^{T-1}\boldsymbol{\epsilon}(t)+\sum_{t=0}^{T-1}\beta\nabla F(\boldsymbol{\epsilon}(t))+\boldsymbol{\epsilon}(T).
\end{equation}
\end{small}
\vspace{-0.1in}

 Here, $\boldsymbol{\theta}(0)$ is the initial model and $\boldsymbol{\theta}(0)-\sum_{t=0}^{T-1}\boldsymbol{g}_t$ captures the ideal FL model after $T$ global rounds. Furthermore,
$\sum_{t=0}^{T-1} \boldsymbol{\epsilon}(t)$ is the accumulated error of all $T$ global rounds, $\sum_{t=0}^{T-1}\beta\nabla F(\boldsymbol{\epsilon}(t))$ is the accumulated gradient distortion across $T$ global rounds, and $\boldsymbol{\epsilon}(T)$ is the error for the final aggregation round.

For the accumulated error term $\sum_{t=0}^{T-1}\boldsymbol{\epsilon}(t)$, we have
\begin{equation}
    \mathbb{E}\left[ \left| \left|\sum_{t=0}^{T-1}\boldsymbol{\epsilon}(t)\right|\right|^2 \right] =\sum_{t=0}^{T-1}\mathbb{E}\left[ \left| \left|\boldsymbol{\epsilon}(t)\right|\right|^2 \right]=Te(t),
\end{equation}

For the accumulated gradient distortion term $\sum_{t=0}^{T-1}\beta\nabla F(\boldsymbol{\epsilon}(t))$, we have
\begin{equation}\label{eta_11}
    \sum_{t=0}^{T-1}\beta\nabla F(\boldsymbol{\epsilon}(t))\leq \beta\eta \left| \left|\boldsymbol{\epsilon}(t)\right|\right|^2= \beta\eta e(t).
\end{equation}

Here, $\eta$ represents the degree to which the gradient-based model update is nonlinearly distorted by the aggregation error $\boldsymbol{\epsilon}(t)$. Its value is related to the Lipschitz continuity property of the gradient of the loss function \cite{zehtabi2024decentralized}.

Combining all the above analyses, we have the accumulated error between the perfect OTA-FL results and the imperfect results with misalignment and time varying channels as:
\vspace{-0.1in}

\begin{small}
\begin{equation} \label{Final_error}
    \mathbb{E} \left[ 
    \left| \left|
    \sum_{t=0}^{T-1}\boldsymbol{\epsilon}(t)\!+\!\!\!\sum_{t=0}^{T-1}\beta\nabla F(\boldsymbol{\epsilon}(t))\!+\!\boldsymbol{\epsilon}(T)
    \right| \right|^2
    \right]\!\! \leq
    (T\!+\!1\!+\!\beta\eta)e(t).
\end{equation}
\end{small}
\vspace{-0.1in}


\vspace{-0.03in}
\section{Simulation Results} \label{sec:simulations}

In this section, we provide simulation results to verify the theoretical error bound we derived in Eq. \eqref{Final_error}. We assume that there are $K=10$ local devices, each with a single antenna. The number of receiving antennas at the central PS is $N=5$. The channel coefficient variance is set to $\sigma_h=0.2$ and the noise variance to $\sigma_z=0.1$. For the phase parameters introduced by the misalignment delay and
time-varying channels, we set them as $\sigma_{\tau}=0.01$ and $\mu_{\tau}=0.1$. The number of subcarriers used in the simulation is $F_{sc}=64$, and the ICI window is $q=2$. We define the MSE loss function as $L(\boldsymbol{w}) = \frac{1}{n} \sum_{i=1}^{n} (y_i - \boldsymbol{x}_i^T \boldsymbol{w})^2$ with $\boldsymbol{w}$ being the model parameter vector. We assume that each data point $\boldsymbol{x}_i$ is drawn independently from a standard Gaussian distribution, i.e., $\boldsymbol{x}_i\backsim \mathcal{N}(0,\mathbf{I}_d)$, and $y_i$ is generated by adding additive Gaussian noise that is independent of $\boldsymbol{x}_i$.
Based on this, the gradient can be expressed as $\nabla L(\mathbf{w}) = \frac{2}{n} \mathbf{X}^T (\mathbf{X}\mathbf{w} - \mathbf{y})$. Each local device utilizes this formulation to compute its local model parameters. These parameters are then transmitted to the central PS over a delay- and time-varying wireless channel, as described in Eq.~\eqref{time_varying_channel}. Upon reception, the central PS aggregates the received local models and broadcasts the updated global model back to all local devices to facilitate the subsequent round of model updates \cite{setting1}.

\begin{figure}[t] 
\centering
\includegraphics[width=5.8cm]{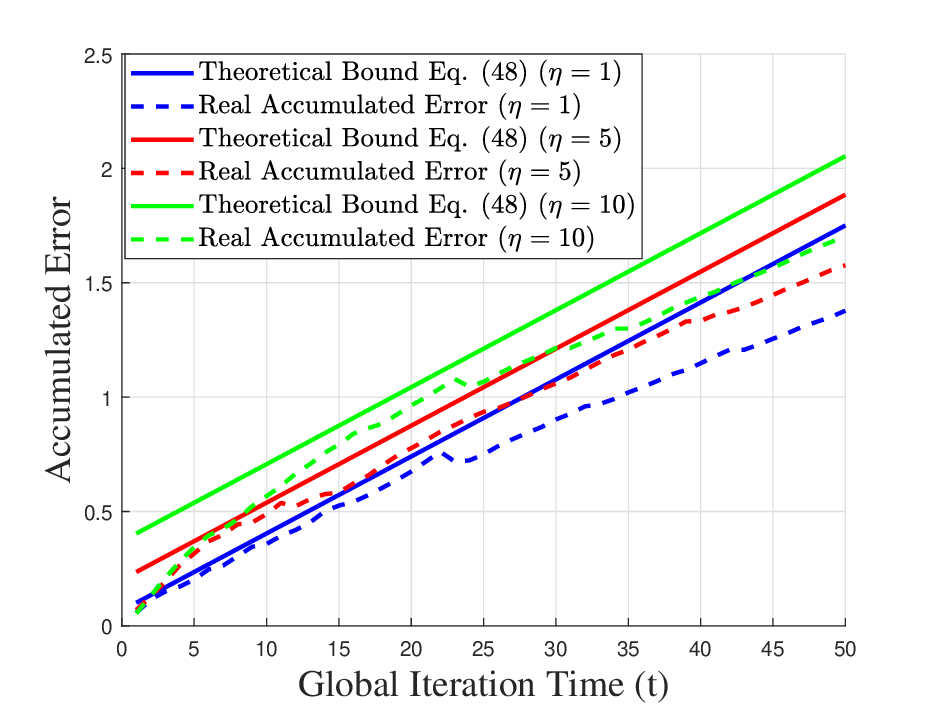}
\caption{Comparison between the derived analytic error bound in Eq. \eqref{Final_error} and the real accumulated error through multi-round OTA-FL with different $\eta$ values.}
\label{eta_value}
\end{figure}

\begin{figure}[t!]
\centering
\includegraphics[width=5.8cm]{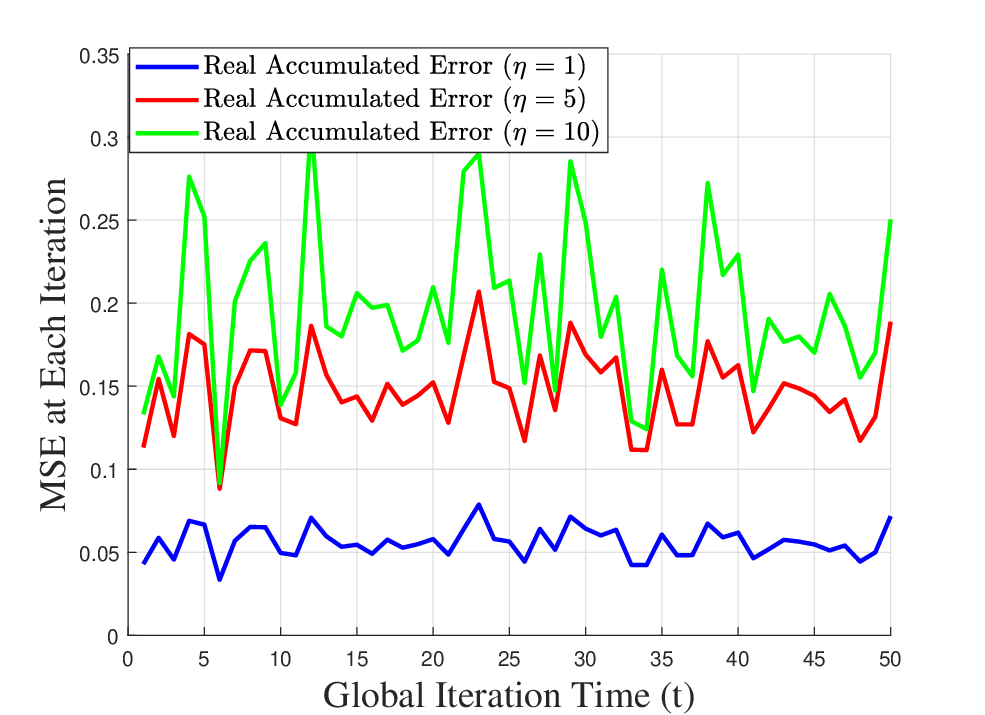}
\caption{The actual error of each global round for practical OTA-FL with different $\eta$ values.}
\label{Simulation_2}
\end{figure}

We first compare the derived analytic error bound in Eq. \eqref{Final_error} and the real accumulated error through multi-round OTA-FL with different $\eta$ values. Specifically, the $\eta$ parameter we introduced in Eq. \eqref{eta_11} represents the nonlinear gradient distortion introduced by the aggregation errors.  For a $L_{Lip}$-smooth loss-function $F(\cdot)$, $\eta$ should satisfy $\eta \leq L_{Lip}^2$. Here, we set $\eta=1, 5, 10$. 
In Fig. \ref{eta_value}, we see that the theoretical error bound we derived in Eq. \eqref{Final_error} and the actual accumulated OTA-FL errors have the same growth trend. The gap between the theoretical error bound and the practical error bound is relatively small, which verifies that the theoretical error we obtained can be well utilized to model the actual error. Another notable observation is that as the global rounds progress, the gap between the theoretical bound and the practical error becomes larger. This is because when deriving the theoretical bound, we always use the worst case scenario to serve as upper bound. Whereas in practical OTA-FL, as shown in Fig. \ref{Simulation_2}, the errors follow a more realistic stochastic process as some errors from different rounds may cancel each other out.

\begin{figure}[t!]
\centering
\includegraphics[width=5.8cm]{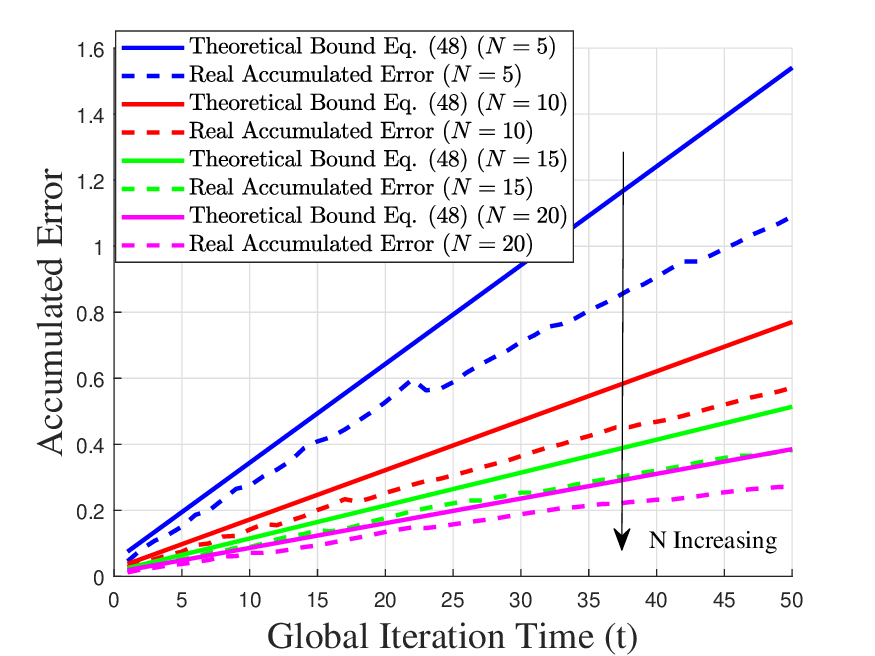}
\caption{Comparison between the derived analytic error bound and the real accumulated error through multi-round OTA-FL under different numbers of receive antennas $N$ at the central PS.}
\label{Simulation_3}
\end{figure}

In Fig. \ref{Simulation_3}, we simulate the comparisons between the derived analytic error bound and the real accumulated error through multi-round OTA-FL with different numbers of receiving antennas ($N$) at the central PS. Here, $\eta$ is set to $5$. We can see that as the number of $N$ increases, the accumulated error decreases, which is consistent with our theoretical analysis and validates the accuracy of our theoretical analysis.

\vspace{-0.04in}
\section{Conclusion} \label{sec:conclusion}

In this paper, we derived a closed-form expression for a single-round global model update, then derived the accumulated error across learning rounds. We also provided numerical simulations to
confirm the accuracy of our derived analysis.

\bibliographystyle{IEEEtran}
\footnotesize
\bibliography{ref}
\end{document}